\documentclass[showpacs,preprintnumbers,amsmath,amssymb]{revtex4}
\newtheorem{theorem}{Theorem}

\newtheorem{lemma}{Lemma}

\newcommand{\BlackBox}{\rule{1.5ex}{1.5ex}}  
\newenvironment{proof}{\par\noindent{\bf Proof:\
}}{\hfill\BlackBox\\[2mm]}

\usepackage{graphicx}
\usepackage{dcolumn}
\usepackage{bm}

\begin{document}



\title{Relation between 
information and disturbance 
in quantum key distribution protocol with classical Alice}

\author{Takayuki Miyadera}

\affiliation{Research Center for Information Security, 
\\
National Institute of 
Advanced Industrial Science and Technology, 
\\
1-1-1 Umezono, 
Tsukuba, Ibaraki 305-8561, Japan
\\
miyadera-takayuki@aist.go.jp}


\begin{abstract}
The ``semiquantum" key distribution protocol introduced by Zou {\em et al.} 
is examined. 
The protocol while using two-way quantum communication 
requires only Bob to be fully quantum.  
We derive a trade-off inequality between  
information gained by Eve and the disturbance observed by legitimate users. 
It guarantees that Eve cannot obtain  
large information if the disturbance 
is sufficiently small. 
\end{abstract}


\maketitle
\section{Introduction}
Recently, several ``semiquantum" key distribution protocols 
were proposed\cite{BKM,Semi,Zou}.  
In contrast to the common quantum key distribution protocols such as BB84, 
one of the parties in these protocols uses only classical operations.
In their pioneering work, 
Boyer, Kenigsberg and Mor introduced\cite{BKM} a two-way 
semiquantum key distribution 
 protocol using four states. 
Zou, Qiu, Li, Wu, and Li 
derived\cite{Zou} its simplification that requires only one state. 
These protocols are interesting, because they give insights into the 
necessary conditions for achieving secure communication. 
In return for the merit that the protocols need only one quantum party, they 
use two-way quantum communication channels. 
This makes the security proof difficult. 
In fact, only 
the robustness of the protocols has been proved so far\cite{BKM,Semi,BoyerMor,ZQ,photon}. 
The robustness of the protocols suggests  
that
information gained by Eve inevitably 
disturbs the communication between Alice and Bob. 
While this robustness is necessary for the security of the protocols,  
as the no-cloning theorem was in the BB84 protocol, 
the next important step should be taken for showing a quantitative 
trade-off relationship between the information gained by Eve 
and the disturbance observed by legitimate users\cite{Biham,Hayashi}. 
This type of relationship in the BB84 protocol is called the  
information-disturbance theorem\cite{Boykin,Miyainfo}. 
\par
In this paper, we derive such a trade-off relationship in  
the protocol introduced by Zou {\em et al.}\cite{Zou}. In this protocol, 
the existence of Eve is noticed by performing 
two error-checking procedures. The inequality we derive 
relates the amount of information gained by Eve to these 
error probabilities.  
\section{Formulation and results}
\subsection{Formulation}
The protocol given by Zou {\em et al.}\cite{Zou} runs as follows. 
Bob sends Alice $N$ qubits 
each in the state $|+\rangle:=\frac{1}{\sqrt{2}}(|0\rangle +|1\rangle)$ 
and keeps all qubits he receives back from her in a quantum memory. 
After confirming the receipt of all qubits by Bob, 
Alice publicly announces  
which qubits she reflected (without disturbing them); 
Bob then checks that he received $|+\rangle$ and not $|-\rangle
:=\frac{1}{\sqrt{2}}(|0\rangle -|1\rangle)
$ on 
those positions (CTRL). For the (SIFT) qubits measured by Alice in the 
standard (classical) $\{|0\rangle, |1\rangle\}$ basis, a sample is chosen to be checked 
for errors (TEST). The remaining SIFT bits 
serve for obtaining a final key via error correction and 
privacy amplification. 
\par
Instead of this full protocol, we treat its toy version using a qubit 
without the public discussion. 
This protocol including Eve's attack is described as follows. 
We consider two situations: CTRL and SIFT. 
In both situations, Bob first sends a qubit to Alice 
in the state $|+\rangle \in {\cal H}:={\bf C}^2$. 
Eve makes the qubit interact with her apparatus ${\cal K}$ by 
a unitary operation $V:{\cal H} \otimes {\cal K} \to {\cal H}\otimes {\cal K}$.  
The whole state evolves into 
\begin{eqnarray*}
|\Psi\rangle:=V|+\rangle \otimes |\Omega\rangle,
\end{eqnarray*}
 where $|\Omega\rangle$ denotes 
the initial state of Eve's apparatus.
\par
In the case of CTRL, Alice reflects the qubit without disturbing it. 
Eve again makes the qubit sent from Alice to Bob interact with her apparatus. 
It is described by a unitary operation $U:{\cal H}\otimes {\cal K} 
\to {\cal H} \otimes {\cal K}$. 
The whole state after the interaction is thus described 
as $U|\Psi\rangle=UV|+\rangle \otimes |\Omega\rangle$. 
Bob measures a projection-valued measure (PVM) $X=
\{X_{+},X_-\}:=
\{|+\rangle \langle +|\otimes {\bf 1}_{\cal K}, |-\rangle \langle -|\otimes {\bf 1}_{\cal K}\}$ 
to check whether the state is in $|+\rangle$. 
We define $P_{CTRL}$ by 
$P_{CTRL}:=\langle \Psi| U^*X_- U|\Psi\rangle$, 
which is an error probability in CTRL. 
\par
In the case of SIFT, after receiving a qubit, Alice measures a 
PVM $Z=\{Z_0,Z_1\}:=
\{|0\rangle \langle 0|\otimes {\bf 1}_{\cal K}, 
|1\rangle \langle 1|\otimes {\bf 1}_{\cal K} \}$. 
The probability for obtaining $z\in \{0,1\}$ is 
calculated as $p_{SIFT}^A(z):=\langle \Psi |Z_z |\Psi\rangle$. 
The state after the measurement is changed according to the von Neumann-L\"uders postulate. 
If $z$ is obtained, the whole state becomes 
\begin{eqnarray*}
\sigma_z:=\frac{Z_z |\Psi\rangle \langle \Psi |
Z_z }
{p_{SIFT}^A(z)}.
\end{eqnarray*} 
Alice sends the qubit back to Bob. 
Also in this case, Eve makes the qubit interact with her apparatus by using $U$. 
The whole system thus becomes 
$U\sigma_z U^*$. After receiving the qubit, Bob checks the state by 
measuring $Z$. The (conditional) probability for obtaining $z'\in \{0,1\}$ 
when Alice's outcome is $z$ is represented as 
$p^{B|A}_{SIFT}(z'|z):=\mbox{tr}(U \sigma_z U^* Z_{z'})$. 
Using these quantities, we define an error probability in SIFT by 
$P_{SIFT}:=p^{B|A}_{SIFT}(1|0)p^A_{SIFT}(0)+p^{B|A}_{SIFT}(0|1)p^A_{SIFT}(1)$. 
This quantity is represented as 
$P_{SIFT}=\langle \Psi |Z_0 U^* Z_1 U Z_0|\Psi\rangle + \langle \Psi | Z_1 U^* Z_0 U Z_1 |\Psi\rangle$. 
Eve's purpose is to know the outcome obtained by Alice. 
Let us denote the state possessed by Eve after the two-way 
quantum communication when Alice 
obtains $z$ 
in SIFT by $\rho_z$. 
It is represented as 
\begin{eqnarray*}
\rho_z:=\mbox{tr}_{{\cal H}}(U\sigma_z U^*),
\end{eqnarray*}
where $\mbox{tr}_{{\cal H}}$ is the partial trace 
over ${\cal H}$. 
Eve measures a positive-operator-valued measure (POVM) $E=\{E_e\}$ which acts only on ${\cal K}$ for extracting information.  
That is, each $E_e$ can be represented as $E_e={\bf 1}_{\cal H} \otimes \hat{E_e}$ 
by using some $\hat{E_e}$. We denote by 
$p^{E|A}_{SIFT}(e|z)$ the probability for obtaining an outcome $e$ when Alice 
obtains $z$. It is represented as
$p^{E|A}_{SIFT}(e|z)=\mbox{tr}(\rho_z E_e)$. 
We denote by $p^{AE}_{SIFT}(z,e)$ the joint probability representing 
the case that Alice obtains $z$ and Eve obtains $e$. 
This quantity is calculated as 
$p^{AE}_{SIFT}(z,e)=p^{E|A}_{SIFT}(e|z)p^A_{SIFT}(z)
=\langle \Psi |Z_z U^* E_e U Z_z|\Psi\rangle$. 
In addition, 
the probability for obtaining $e$ is calculated as 
$p^E_{SIFT}(e):=\sum_z p^{E|A}_{SIFT}(e|z)p^A_{SIFT}(z)$. 
The information gained by Eve is characterized by the 
mutual information, which is defined by 
\begin{eqnarray*}
I(A:E)=H(A)+H(E)-H(A,E),
\end{eqnarray*} 
where $H(A):=- \sum_z p^A_{SIFT}(z) \log_2 p^A_{SIFT}(z)$, 
$H(E):=- \sum_e p^E_{SIFT}(e) \log_2 p^E_{SIFT}(e)$ 
and 
$H(A,E):=-\sum_{z,e}p^{AE}_{SIFT}(z,e)
\log_2 p^{AE}(z,e)$. 
Eve has two chances to make her apparatus interact 
with the qubit. It is obvious that each interaction 
can help her obtain information. For instance, Eve can 
have an entangled state between her apparatus and 
the qubit sent to Alice by using $V$. Although it brings her 
information, this interaction leaves its trace behind by 
disturbing the state. 
Our aim in this paper is to derive a trade-off inequality 
that bounds $I(A:E)$ by $P_{CTRL}$ and $P_{SIFT}$ for 
general attacks in which both $U$ and $V$ are arbitrary. 
\subsection{Relation between information and disturbance}
The following is our main theorem. 
\begin{theorem}
The information gained by Eve can be bounded from above as
\begin{eqnarray*}
I(A:E)&\leq& 
2 \sqrt{P_{CTRL}+6P_{SIFT}^{1/4}
}, 
\end{eqnarray*}
where $P_{SIFT}$ and $P_{CTRL}$ are the error probabilities defined above. 
\end{theorem}
This theorem generalizes the robustness result. 
In fact, if we put $P_{CTRL}=P_{SIFT}=0$ in the above inequality, 
$I(A:E)=0$ follows. That is, information gained by Eve inevitably 
causes disturbance. 
Moreover, the theorem guarantees that information gained by Eve is 
small if both of the probabilities 
$P_{CTRL}$ and $P_{SIFT}$ 
are sufficiently small. 
\par
We employ two lemmas to prove our main theorem. 
The following lemma is employed to bound the mutual information 
by a quantity that is easier to 
treat. 
\begin{lemma}\label{lemma1}
Let $X$ and $Y$ be random variables. 
Suppose that $X$ takes a value in $\{0,1\}$. Denote by $p^{XY}(x,y)$ the  
joint probability representing the case that $X$ takes $x$ and $Y$ takes $y$. 
The mutual information between $X$ and $Y$ is bounded as 
\begin{eqnarray*}
I(X:Y)\leq \sqrt{
1-4 \left( \sum_y p^{XY}(0,y)^{1/2}p^{XY}(1,y)^{1/2}\right)^2
}.
\end{eqnarray*}
\end{lemma}
\begin{proof}
The proof is the slightest modification of Theorem 1 in Ref.\cite{Fuchs}. 
Let us denote by $p^X(x)$ and $p^Y(y)$ the marginal probabilities 
with their apparent notations, 
and by 
$p^{X|Y}(x|y)$ the conditional probability 
defined by $p^{X|Y}(x|y)=\frac{p^{XY}(x,y)}{p^{Y}(y)}$. 
The mutual information can be written as 
\begin{eqnarray*}
I(X:Y)=H(X)-H(X|Y), 
\end{eqnarray*}
where 
$
H(X):=-\sum_x p^{X}(x)\log_2 p^X(x)
$
and 
$
H(X|Y):=-\sum_y p^{Y}(y)\sum_x p^{X|Y}(x|y)\log_2 p^{X|Y}(x|y). 
$
Because $H(X)\leq 1$ and  
$-\sum_x p^{X|Y}(x|y) \log_2 p^{X|Y}(x|y) 
\geq 2 \min\{p^{X|Y}(0|y), p^{X|Y}(1|y)\}$ hold, it holds that  
\begin{eqnarray*}
I(X:Y)\leq 1-2\sum_y p^Y(y) \min \{p^{X|X}(0|y), p^{X|Y}(1|y)\}. 
\end{eqnarray*}
Using $p^{X|Y}(0|y)+p^{X|Y}(1|y)=1$, we obtain 
$-2 \min\{p^{X|Y}(0|y), p^{X|Y}(1|y)
\}=-1 +\left| p^{X|Y}(0|y) -p^{X|Y}(1|y) \right|$. 
Thus it holds that 
\begin{eqnarray*}
I(X:Y)
\leq \sum_y p^Y(y) \left| p^{X|Y}(0|y) -p^{X|Y}(1|y) \right|
=\sum_y \left|
p^{XY}(0,y)-p^{XY}(1,y)
\right|. 
\end{eqnarray*}  
The right-hand side of this inequality can be bounded as follows:
\begin{eqnarray*}
&&
\sum_y \left|
p^{XY}(0,y)-p^{XY}(1,y)
\right|
\\
&=&\sum_y \left| 
\sqrt{p^{XY}(0,y)}-\sqrt{p^{XY}(1,y)}\right|
\left| \sqrt{p^{XY}(0,y)}+\sqrt{p^{XY}(1,y)}
\right|
\\
&\leq&
\left(
\sum_y 
\left(
\sqrt{p^{XY}(0,y)}-\sqrt{p^{XY}(1,y)}\right)^2
\sum_y 
\left( \sqrt{p^{XY}(0,y)}+\sqrt{p^{XY}(1,y)}
\right)^2
\right)^{1/2}
\\
&=&
\left(1-2\sum_y p^{XY}(0,y)^{1/2} p^{XY}(1,y)^{1/2}
\right)^{1/2}
\left(1+2\sum_y p^{XY}(0,y)^{1/2} p^{XY}(1,y)^{1/2}
\right)^{1/2}
\\
&=&
\left(
1-4 \left(\sum_y p^{XY}(0,y)^{1/2} p^{XY}(1,y)^{1/2}\right)^2
\right)^{1/2}, 
\end{eqnarray*}
where we used the Cauchy-Schwarz inequality. 
\end{proof}
The following lemma plays an important role in 
relating probabilities in SIFT and CTRL with each other. 
\begin{lemma}\label{lemmakey}
For any (possibly unnormalized) vectors $|\phi_0\rangle, 
|\phi_1\rangle \in {\cal H}\otimes {\cal K}$, 
any bounded operator $X$ acting only on ${\cal H}$, any POVM $E=\{E_e\}$ acting only on 
${\cal K}$, it holds that 
\begin{eqnarray}
|\langle \phi_0 |X |\phi_1\rangle |
\leq \Vert X\Vert \sum_e  \langle \phi_0 |E_e|\phi_0\rangle^{1/2}
\langle \phi_1 |E_e|\phi_1 \rangle^{1/2}, 
\label{star}
\end{eqnarray}
where $\Vert\cdot \Vert$ is an operator norm defined by 
$\Vert X\Vert:=\sup_{\phi\neq 0}\frac{\Vert X|\phi\rangle \Vert}{\Vert |\phi\rangle\Vert}$. 
\end{lemma}
\begin{proof}
Using the commutativity between $E_e^{1/2}$ and $X$, we obtain 
\begin{eqnarray*}
|\langle \phi_0 |X|\phi_1 \rangle |
=|\sum_e \langle \phi_0 |XE_e |\phi_1 \rangle |
=|\sum_e \langle \phi_0 |E_e^{1/2} X E_e^{1/2} |\phi_1 \rangle |.
\end{eqnarray*}
We further obtain 
\begin{eqnarray*}
|\langle \phi_0|X|\phi_1 \rangle |
&\leq& \sum_e |\langle \phi_0 |E_e^{1/2}XE_e^{1/2}|\phi_1 \rangle |
\\
&\leq & 
\sum_e \langle \phi_0 |E_e |\phi_0\rangle^{1/2}
\langle \phi_1 |E_e^{1/2}X^* X E_e^{1/2}|\phi_1 \rangle^{1/2}
\\
&\leq &
\sum_e \langle \phi_0 |E_e|\phi_0 \rangle^{1/2}
\langle \phi_1|E_e |\phi_1\rangle^{1/2}\Vert X\Vert, 
\end{eqnarray*}
where we used the Cauchy-Schwarz inequality to derive the second line 
and the definition of the operator norm to derive the third line. 
\end{proof}
\begin{proof}(Proof of Theorem 1)
We apply Lemma \ref{lemma1} to $p^{AE}_{SIFT}(z,e)$ in order to 
bound $I(A:E)$. 
To bound $p^{AE}_{SIFT}(z,e)$ by $P_{SIFT}$ and $P_{CTRL}$, 
we compare this quantity with another probability 
$p_0(z,e)$ defined by 
$p_0(z,e):=\langle \Psi |U^* Z_z E_e U |\Psi\rangle$. 
Using $Z_z +Z_{z\oplus 1}={\bf 1}$, we obtain  
\begin{eqnarray*}
U Z_z &=& Z_z U Z_z + Z_{z\oplus 1} U Z_z \\ 
&=& Z_z U +C_z, 
\end{eqnarray*}
where $C_z := Z_{z \oplus 1} U Z_z -Z_z U Z_{z\oplus 1}$. 
Thus it holds that   
\begin{eqnarray*}
p_{SIFT}^{AE}(z,e) 
&=&\langle \Psi | (U^* Z_z + C^*_z) E_e (Z_z U + C_z)|\Psi\rangle
\\
&=& p_0 (z,e) +\langle \Psi |U^* Z_z E_e C_z |\Psi \rangle 
+\langle \Psi | C^*_z Z_z E_e U |\Psi \rangle 
+\langle \Psi |C^*_z E_e C_z |\Psi\rangle. 
\end{eqnarray*}
We obtain 
\begin{eqnarray*}
&&\left| p_{SIFT}^{AE}(z,e)- p_0 (z,e)\right|
= \left| 
\langle \Psi |U^* Z_z E_e C_z |\Psi \rangle 
+\langle \Psi | C^*_z Z_z E_e U |\Psi \rangle 
+\langle \Psi |C^*_z E_e C_z |\Psi\rangle
\right|
\\
&\leq &
\langle \Psi|U^* Z_z E_e U|\Psi\rangle^{1/2}\langle \Psi |C_z^* E_e C_z|\Psi\rangle^{1/2}
+\langle \Psi |C_z^* E_e C_z |\Psi \rangle^{1/2} 
\langle \Psi | U^* Z_z E_e U |\Psi \rangle^{1/2} 
\\
&&
+\langle \Psi | C_z^* E_e C_z|\Psi \rangle 
\\
&=& 2 p_0 (z,e)^{1/2} \langle \Psi | C_z^* E_e C_z|\Psi \rangle^{1/2} 
+ \langle \Psi |C_z^* E_e C_z|\Psi \rangle, 
\end{eqnarray*}
where we used the triangular inequality and the Cauchy-Schwarz inequality. 
Because $|a-b|\leq c$ implies $|\sqrt{a}-\sqrt{b}|\leq \sqrt{c}$ for positive $a,b$ and $c$, 
it holds that 
\begin{eqnarray}
\left| \sqrt{p_{SIFT}^{AE}(z,e)}- \sqrt{p_0 (z,e)}\right|
\leq \left(
2 p_0 (z,e)^{1/2} \langle \Psi | C_z^* E_e C_z|\Psi \rangle^{1/2} 
+ \langle \Psi |C_z^* E_e C_z|\Psi \rangle
\right)^{1/2}. 
\label{sqrt}
\end{eqnarray}
We apply Lemma \ref{lemmakey} to 
$|\phi_0\rangle =Z_0 U|\Psi \rangle$, $|\phi_1 \rangle =Z_1 U|\Psi\rangle$ 
and $X=|0\rangle \langle 1|\otimes {\bf 1}_{\cal K}$. 
The left-hand side of (\ref{star}) can be bounded as 
\begin{eqnarray}
|\langle \phi_0 |X |\phi_1 \rangle |
&=&| \langle \Psi | U^* (|0\rangle \langle 1|\otimes {\bf 1}_{{\cal K}}) U |\Psi\rangle |
\nonumber 
\\
&\geq & 
\frac{\langle \Psi | U^* (|0\rangle \langle 1| \otimes {\bf 1}_{{\cal K}} ) U |\Psi \rangle 
+ \langle \Psi | U^* (|1\rangle \langle 0 | \otimes {\bf 1}_{{\cal K}}) U |\Psi \rangle }{2}
\nonumber 
\\
&=& \frac{1}{2}- \langle \Psi | U^* X_- U |\Psi \rangle 
\nonumber 
\\
&=& \frac{1}{2}- P_{CTRL}.
\label{sahen}
\end{eqnarray}
The right-hand side of (\ref{star}) becomes
\begin{eqnarray}
&& 
\Vert X\Vert \sum_e  \langle \phi_0 |E_e|\phi_0\rangle^{1/2}
\langle \phi_1 |E_e|\phi_1 \rangle^{1/2}
=\sum_e p_0(0,e)^{1/2} p_0(1.e)^{1/2} 
\nonumber 
\\
&\leq &
\sum_e \left(
p^{AE}_{SIFT}(0,e)^{1/2}
+\left(2 p_0(0,e)^{1/2} \langle \Psi | C_0^* E_e C_0|\Psi \rangle^{1/2}
+\langle \Psi |C_0^* E_e C_0|\Psi \rangle \right)^{1/2}
\right)
\nonumber \\
&&
\cdot
\left(
p^{AE}_{SIFT}(1,e)^{1/2}
+\left(2 p_0(1,e)^{1/2} \langle \Psi | C_1^* E_e C_1|\Psi \rangle^{1/2}
+\langle \Psi |C_1^* E_e C_1|\Psi \rangle \right)^{1/2}
\right),  
\label{tochu}
\end{eqnarray}
where we used (\ref{sqrt}). 
By using the Cauchy-Schwarz inequality, 
we can further bound the above inequality as
\begin{eqnarray*}
&&(\ref{tochu})
\leq 
\sum_e p^{AE}_{SIFT}(0,e)^{1/2}p^{AE}_{SIFT}(1,e)^{1/2}
\\
&&
+p^A_{SIFT}(0)^{1/2}
\left(
\sum_e (2p_0(1,e)^{1/2}\langle \Psi |C_1^* E_e C_1|\Psi\rangle^{1/2}
+\langle \Psi |C_1^* E_e C_1|\Psi\rangle)
\right)^{1/2}
\\
&&
+
p^A_{SIFT}(1)^{1/2}
\left(
\sum_e (2p_0(0,e)^{1/2}\langle \Psi |C_0^* E_e C_0|\Psi\rangle^{1/2}
+\langle \Psi |C_0^* E_e C_0|\Psi\rangle)
\right)^{1/2}
\\
&&
+\left(
\sum_e (2p_0(1,e)^{1/2}\langle \Psi |C_1^* E_e C_1|\Psi\rangle^{1/2}
+\langle \Psi |C_1^* E_e C_1|\Psi\rangle)
\right)^{1/2}
\\
&&\times
\left(
\sum_e (2p_0(0,e)^{1/2}\langle \Psi |C_0^* E_e C_0|\Psi\rangle^{1/2}
+\langle \Psi |C_0^* E_e C_0|\Psi\rangle)
\right)^{1/2}.
\end{eqnarray*}
The terms $
\sum_e (2p_0(0,e)^{1/2}\langle \Psi |C_0^* E_e C_0|\Psi\rangle^{1/2}
+\langle \Psi |C_0^* E_e C_0|\Psi\rangle)$
and $
\sum_e (2p_0(1,e)^{1/2}\langle \Psi |C_1^* E_e C_1|\Psi\rangle^{1/2}
+\langle \Psi |C_1^* E_e C_1|\Psi\rangle)$ are bounded as 
\begin{eqnarray*}
&&
\sum_e (2p_0(0,e)^{1/2}\langle \Psi |C_0^* E_e C_0|\Psi\rangle^{1/2}
+\langle \Psi |C_0^* E_e C_0|\Psi\rangle)
\\
&&
\leq 2p^A_0(0)^{1/2} \langle \Psi |C_0^* C_0|\Psi \rangle^{1/2}
+\langle \Psi |C_0^* C_0|\Psi \rangle
\\
&&
\sum_e (2p_0(1,e)^{1/2}\langle \Psi |C_1^* E_e C_1|\Psi\rangle^{1/2}
+\langle \Psi |C_1^* E_e C_1|\Psi\rangle)
\\
&&
\leq 2p^A_0(1)^{1/2} \langle \Psi |C_1^* C_1|\Psi \rangle^{1/2}
+\langle \Psi |C_1^* C_1|\Psi \rangle, 
\end{eqnarray*}
where 
$p_0^A(z):=\sum_e p_0(z,e)$ and 
we used the Cauchy-Schwarz inequality and the relation 
$\sum_e E_e={\bf 1}_{{\cal K}}$. 
Because $C_z^* C_z 
=Z_z U^* Z_{z\oplus 1} U Z_z + Z_{z\oplus 1} U^* Z_z U Z_{z\oplus 1}
$ holds, 
we have, for $z=0,1$, 
\begin{eqnarray*}
\langle \Psi |C_z^* C_z|\Psi \rangle 
= P_{SIFT}. 
\end{eqnarray*}
Thus we obtain 
\begin{eqnarray*}
&& 
\Vert X\Vert \sum_e  \langle \phi_0 |E_e|\phi_0\rangle^{1/2}
\langle \phi_1 |E_e|\phi_1 \rangle^{1/2}
\leq \sum_e p^{AE}_{SIFT}(0,e)^{1/2}p^{AE}_{SIFT}(1,e)^{1/2}
\\
&&
+p_{SIFT}^A(0)^{1/2}(2p_0^A(1)^{1/2}P_{SIFT}^{1/2}+P_{SIFT})^{1/2}
\\
&&
+p_{SIFT}^A(1)^{1/2}(2p_0^A(0)^{1/2}P_{SIFT}^{1/2}+P_{SIFT})^{1/2}
\\
&&
+(2p^A_0(1)^{1/2}P_{SIFT}^{1/2}+P_{SIFT})^{1/2}
(2p^A_0(0)^{1/2}P_{SIFT}^{1/2}+P_{SIFT})^{1/2}. 
\\
&\leq &
\sum_e p^{AE}_{SIFT}(0,e)^{1/2}p^{AE}_{SIFT}(1,e)^{1/2}
\\
&&
+(2p^A_0(1)^{1/2}P_{SIFT}^{1/2} 
+2p^A_0(0)^{1/2}P_{SIFT}^{1/2}+2 P_{SIFT})^{1/2}
\\
&& +
\frac{1}{2} 
\left( 
2p^A_0(1)^{1/2}P_{SIFT}^{1/2} 
+2p^A_0(0)^{1/2}P_{SIFT}^{1/2}+2 P_{SIFT}\right), 
\end{eqnarray*}
where we used the Cauchy-Schwarz inequality again. 
Because $P_{SIFT}\leq P_{SIFT}^{1/2}$ and $p^A_0(0), p^A_0(1)\leq 1$ hold, 
it holds that  
\begin{eqnarray}
&&
\Vert X\Vert \sum_e  \langle \phi_0 |E_e|\phi_0\rangle^{1/2}
\langle \phi_1 |E_e|\phi_1 \rangle^{1/2}
\nonumber \\
&\leq & \sum_e p^{AE}_{SIFT}(0,e)^{1/2}p^{AE}_{SIFT}(1,e)^{1/2} 
+
\sqrt{6}P_{SIFT}^{1/4} +3 P_{SIFT}^{1/2}
\nonumber 
\\
&
\leq &
\sum_e p^{AE}_{SIFT}(0,e)^{1/2}p^{AE}_{SIFT}(1,e)^{1/2} + 
 6 P_{SIFT}^{1/4}, 
\label{uhen} 
\end{eqnarray}
where we used $\sqrt{6}+3 < 6$. 
(Although the above inequality can be slightly improved, 
we do not treat it here as it is not important.) 
Thus (\ref{sahen}), (\ref{uhen}) and Lemma \ref{lemmakey} derive 
\begin{eqnarray*}
\frac{1}{2}-P_{CTRL} - 6 P_{SIFT}^{1/4} \leq 
\sum_e p^{AE}_{SIFT}(0,e)^{1/2}p^{AE}_{SIFT}(1,e)^{1/2}.
\end{eqnarray*} 
Now we can apply Lemma \ref{lemma1} to obtain
\begin{eqnarray*}
I(A:E)&\leq& 2 \sqrt{(P_{CTRL}+6P_{SIFT}^{1/4}) -(P_{CTRL}+6P_{SIFT}^{1/4})^2}
\\
&\leq &
2 \sqrt{P_{CTRL}+6P_{SIFT}^{1/4}
}. 
\end{eqnarray*}
\end{proof}
\section{Summary}
In this paper, treating the quantum key distribution 
protocol with classical Alice, we obtained a trade-off relationship between  
information gained by Eve and the disturbance observed by Alice and Bob. 
Our theorem provides a generalization of the robustness result 
obtained thus far. 
Moreover it guarantees that information gained by Eve is 
small if both of two error probabilities observed by the legitimate users 
are sufficiently small. 
Applying the inequality to the full protocol in order to 
examine its security is an important future problem. 
\\
{\bf Acknowledgments:} 
I would like to thank Prof. Hideki Imai for his encouragements.

\end{document}